\documentclass[11pt]{article}
\usepackage{authblk}
\usepackage[margin=1in]{geometry}
\usepackage[T1]{fontenc}
\usepackage{bbm} 
\usepackage{amsthm}
\usepackage{enumerate}
\usepackage{comment}
\usepackage{url}
\usepackage{amsmath}
\usepackage{amssymb}
\usepackage{graphicx}
\graphicspath{{./Figures/}}
\usepackage{caption}
\usepackage{subcaption}
\usepackage{paralist}
\usepackage[active]{srcltx}
\usepackage{algorithm}
\usepackage{algpseudocode}
\usepackage[colorlinks=true,linkcolor=blue,citecolor=blue,filecolor=blue,urlcolor=blue]{hyperref}

\theoremstyle{plain}

\newtheorem{theorem}{Theorem}
\newtheorem{remark}{Remark}
\newtheorem{corollary}{Corollary}
\newtheorem{lemma}{Lemma}
\newtheorem{claim}{Claim}
\newtheorem{definition}{Definition}

\usepackage{color}
\usepackage{soul}

\newcommand{\hlc}[2][yellow]{ {\sethlcolor{#1} \hl{#2}} }
\newcommand{\moti}[1]{\hlc[green]{\textbf{MM}: #1}}

\newcommand{\adi}[1]{\hlc[yellow]{\textbf{AR}: #1}}

\algnewcommand\algorithmicswitch{\textbf{switch}}
\algnewcommand\algorithmiccase{\textbf{case}}
\algnewcommand\algorithmicassert{\texttt{assert}}
\algnewcommand\Assert[1]{\State \algorithmicassert(#1)}

\algdef{SE}[SWITCH]{Switch}{EndSwitch}[1]{\algorithmicswitch\ #1\ \algorithmicdo}{\algorithmicend\ \algorithmicswitch}
\algdef{SE}[CASE]{Case}{EndCase}[1]{\algorithmiccase\ #1}{\algorithmicend\ \algorithmiccase}
\algtext*{EndSwitch}
\algtext*{EndCase}

\newtheorem{observation}[theorem]{Observation}
\newcommand{\INDSTATE}[1][1]{\State\quad}

\newcommand{\eps}{\varepsilon}

\newcommand{\poly}{\text{\emph{poly}}}

\newcommand{\opt}{\text{\textsf{opt}}}
\newcommand{\optfrac}{F^*_{\textit{frac}}}
\newcommand{\optint}{F^*_{\textit{int}}}

\newcommand{\dmin}{d_{\min}}
\newcommand{\dmax}{d_{\max}}
\newcommand{\pmax}{p_{\max}}

\newcommand{\eqdf}{\triangleq}

\newcommand{\alg}{\text{\sf{alg}}}
\newcommand{\cut}{\textit{cut}}

\newcommand{\maxpr}{Max-Pkt-Line}
\newcommand{\maxgrid}{Max-Path-Grid}

\newcommand{\Gst}{G^{st}}

\usepackage{mathtools}
\DeclarePairedDelimiter{\ceil}{\lceil}{\rceil}

\newcommand{\ind}[1]{\mathbbm{1}_{#1}}
\newcommand{\expec}[1]{\mathbf{E} \left[ {#1} \right]}
\newcommand{\expectau}[1]{\mathbf{E}_{\tau} \left[ {#1} \right]}
\newcommand{\prbig}[1]{\mathbf{Pr} \left[ \mathopen{#1} \right]}

\newcommand{\Rrnd}{R_{\textit{rnd}}}
\newcommand{\Rfilter}{R_{\textit{fltr}}}
\newcommand{\Rquad}{R_{\textit{quad}}}
\newcommand{\out}{\mathsf{out}}
\newcommand{\inn}{\mathsf{in}}
\newcommand{\dem}{\mathsf{dem}}
\newcommand{\capa}{\mathsf{cap}}
\newcommand{\vshort}{\ell_{VS}}
\newcommand{\short}{\ell_S}
\newcommand{\medium}{\ell_M}

\newcommand{\ignore}[1]{}

\title{A Constant Approximation Algorithm for Scheduling Packets on Line Networks\thanks{A preliminary version of this paper appeared in the proceedings of ESA 2016~\cite{DBLP:conf/esa/EvenMR16}.}}

\date{}

\author[1]{Guy Even}
\author[2]{Moti Medina}
\author[3]{Adi Ros\'{e}n}
\affil[1]{
 Tel Aviv University\\
  \texttt{guy@eng.tau.ac.il}}
\affil[2]{Ben-Gurion University of the Negev\\
\texttt{medinamo@bgu.ac.il}}
\affil[3]{ CNRS and Universit\'{e} de Paris\\
\texttt{adiro@irif.fr}}

\begin{document} 

\maketitle
\begin{abstract}
  In this paper we improve the approximation ratio for the problem of scheduling packets on  line
  networks with bounded buffers, where the aim is that of   maximizing the throughput.
  Each node in the network has a local buffer of bounded size $B$, and each edge (or
  link) can transmit a limited number, $c$, of packets in every time unit.  The input
  to the problem consists of a set of packet requests, each defined by
    a source node, a destination node, and a release time.
    We denote by $n$ the size of the network. A solution for this problem is
    a schedule that delivers (some of the)
  packets to their destinations without violating the capacity constraints of the
  network (buffers or edges). Our goal is to design an efficient algorithm that computes a
  schedule that maximizes the number of packets that arrive to their
  respective destinations.

  We give a randomized approximation algorithm with constant approximation ratio
  for the case where $B=\Theta(c)$.
   This improves over the  previously best  result of $O(\log^* n)$~\cite{RR}.
  Our improvement is based on a new combinatorial lemma that we prove, stating,
   roughly speaking,
   that if packets are allowed to stay put  in buffers only a limited number of time
   steps, $2d$, where $d$ is the longest source-destination distance of any input packet, then the cardinality of the
   optimal solution is decreased by only a constant factor. This claim was not previously known
    in the directed integral (i.e., unsplittable, zero-one) case, and may find additional applications for
     routing and scheduling algorithms.
\end{abstract}

\paragraph{keywords.} Approximation algorithms,  packet scheduling,
admission control, randomized rounding, linear programming.

\section{Introduction}

In this paper we give an approximation algorithm with an improved approximation
ratio for a network-scheduling problem which has been studied in numerous previous
works in a number of variants (cf.~\cite{AKOR,AKK,AZ,EM14,RosenS11,RR,spaaEvenMP15}). The
problem consists of a directed line-network over nodes $\{0,\ldots,n-1\}$, where each
node $i$ can send packets to node $i+1$, and can also store packets in a local
buffer. The maximum number of packets that can be sent in a single time unit over a given link is
denoted by $c$, and the number of packets each node can store at any given time is
denoted by $B$. An instance of the problem  is further defined by a set  of packets
$r_i=(a_i,b_i,t_i)$, $ 1 \leq i \leq M$, where $a_i$ is the source node of the
packet, $b_i$ is its destination node, and $t_i \geq 1$ is the release time of the packet
at vertex $a_i$.  The goal is that of maximizing the number of packets
that reach their respective destinations without violating the links or the buffers capacities.
We give a  randomized approximation algorithm for that problem,
which has a {\em constant} approximation ratio for the case of $B=\Theta(c)$,
improving upon the previous $O(\log^*n)$ approximation ratio given
in~\cite[Theorem 3]{RR}.

Key to our algorithm is a combinatorial lemma (Lemma~\ref{lemma:bounded path length})  which states the following.  Consider a set of packets such that all source-destination distances are
bounded from above by some $d$. The throughput of an optimal solution in which
every  packet $r_i$ must
reach its destination no later than time $t_i+2d$ is an $\Omega(B/c)$-fraction of the throughput of the
unrestricted optimal throughput. This lemma  plays a crucial role in our algorithm, and we
believe that it may find additional applications for scheduling and routing algorithms in networks.
We emphasize that the fractional version of a similar property, i.e., when packets are unsplittable and one accrues a benefit also from the delivery of partial packets, presented first in~\cite{AZ}, does not
imply the integral version that we prove here.

We emphasize  that the problem studied in the present paper, namely, maximizing the throughput on a network with
bounded buffers, has resisted substantial efforts in its (more applicable) distributed, online setting, even for the
simple network of a directed line. Indeed, even the question whether or not there exists
 a constant competitive online distributed algorithm for that problem on the line network
remains unanswered  at this point. We therefore study here the   offline setting with the
hope that, in addition to its own interest, results and ideas from this setting will
contribute to progress on the distributed  problem.
\subsection{Related Work}
The problem of scheduling packets so as to maximize the throughput (i.e., maximize the number of packets that
 reach their destinations) in a network with bounded buffers was first considered in~\cite{AKOR}, where this
 problem is studied for various types of networks in the distributed  setting. The results
 in that paper, even for the simple network of a directed line, were far from tight but no substantial
 progress has been made since on the realistic, distributed and online, setting.  This has motivated the
  study of this problem in easier settings, as a first step towards solving the realistic, possibly applicable,
  scenario.

  Angelov et al.~\cite{AKK}  give centralized online randomized algorithms for the line network,
  achieving an $O(\log^3 n)$-competitive ratio.  Azar and Zachut~\cite{AZ} improved the randomized competitive ratio to $O(\log^2n)$ which was later improved by Even and Medina~\cite{icalpEvenM10,EM14} to $O(\log n)$.
  A deterministic $O(\log^5n)$-competitive algorithm was given in~\cite{spaaEvenM11,EM14}, which was later improved
  in~\cite{spaaEvenMP15} to $O(\log n)$ if buffer and link capacities are not very small (not smaller than $5$).

   The related problem of
   maximizing the throughput when packets have deadlines (i.e., a packet is counted towards the quality of the solution only if it arrives to its destination before a known deadline) on line network with unbounded input queues is known to be NP-hard~\cite{ARSU}. The
   same problem in a certain variant of the setting, where the input queues are bounded, is shown in~\cite{RR}
   to have an $O(\log^*n)$-approximation randomized algorithm. The setting in the present paper is the same setting as the one
   of the latter paper, and the results of~\cite{RR}
     immediately give an $O(\log^*n)$-approximation randomized algorithm for the problem and setting we
     study in the present paper.

\section{Preliminaries}
\label{sec:problem}

\subsection{Model and problem statement}
We consider the standard model of synchronous store-and-forward packet routing
networks~\cite{AKOR,AKK,AZ}. The network is modeled by a directed path over $n$
vertices. Namely, the network is a directed graph $G=(V,E)$, where
$V=\{0,\ldots,(n-1)\}$ and there is a directed edge from vertex $u$ to vertex $v$ if
$v=u+1$.  The network resources are specified by two positive integer parameters $B$
and $c$ that describe, respectively, the local buffer capacity of every vertex and
the capacity of every edge. In every time step, at most $B$ packets can be
stored in the local buffer of each vertex, and at most $c$ packets can be transmitted
along each edge.

The input consists of a set of packet requests $R=\{r_i\}_{i=1}^M$. A packet
request is specified by a $3$-tuple $r_i=(a_i,b_i,t_i)$, where $a_i\in V$ is the
\emph{source node} of the packet, $b_i\in V$ is its \emph{destination node}, and
$t_i\in \mathbb{N}$ is the release time of the packet
at vertex $a_i$. Note that $b_i>a_i$, and
 $r_i$ is ready to leave $a_i$ in time step $t_i$.

A solution is a schedule $S$. For each request $r_i$, the schedule $S$ specifies a
sequence $s_i$ of transitions that packet  $r_i$ undergoes.  A \emph{rejected}
request $r_i$ is simply discarded at time $t_i$, and no further treatment is required
(i.e., $s_i=\{\textit{reject}\}$). An \emph{accepted} request $r_i$ is delivered from
$a_i$ to $b_i$ by a sequence $s_i$ of actions, where each action is either ``store''
or ``forward''. Consider the packet of request $r_i$. Suppose that in the beginning of time step $t$ the
packet is in vertex $v$ (a packet injected at node $v$ at time  $t$ is considered to be at $v$ at the
beginning of time step $t$) .  A store action means that the packet is stored in the
buffer of $v$, consumes one buffer unit of $v$ at time step $t$,
and will still be in vertex $v$ at the beginning of time step $t+1$. A forward action
means that the packet is transmitted to vertex $v+1$, consumes the one unit of ``bandwidth'' of  the
edge between $v$ and $v+1$ at time $t$, and will be in vertex $v+1$ at the beginning of
time step $t+1$. The packet of request $r_i$ reaches its destination $b_i$ after
exactly $b_i-a_i$ forward steps. Once a packet reaches its destination, it is removed
from the network and it no longer consumes any of the network's resources.

\medskip\noindent
A schedule must satisfy the following constraints:
\begin{enumerate}
\item The \emph{buffer capacity constraint} asserts that at any time step $t$, and in
  every vertex $v$, at most $B$ packets are stored in $v$'s buffer.
\item The \emph{link capacity constraint} asserts that at any step $t$, at most $c$
  packets are transmitted along each edge.
\end{enumerate}

The \emph{throughput} of a schedule $S$ is the number of accepted requests.  We
denote the throughput of a schedule $S$ by $|S|$. As opposed to online algorithms,
there is no point in, and one can avoid,  using network resources for  a certain packet unless that packet reaches its
destination. Namely, a packet that is not rejected and does not reach its destination
only consumes network resources without any benefit. Hence, without loss of generality,
we can assume, as we do in the above definition of a schedule, that every packet that is not rejected reaches its destination.

We consider the offline optimization problem of finding a schedule that maximizes the
throughput.
By
\emph{offline} we mean that the algorithm receives all requests in
advance.\footnote{The number of requests $M$ is finite and known in the offline
  setting. This is not the case in the online setting in which the number of requests
  is not known in advance and may be unbounded.}  By \emph{centralized} we mean that
all the requests are known in one location where the algorithm is executed.  Let
$\opt(R)$ denote a schedule of maximum throughput for the set of requests $R$.  Let
$\alg(R)$ denote the schedule computed by $\alg$ on input $R$. We say that the
approximation ratio of a scheduling algorithm \alg\ is $c\geq 1$ if
$\forall R: |\alg(R)| \geq \frac 1c \cdot |\opt(R)|$.
For a randomized algorithm we say that the expected approximation ratio is $c$ if
$\forall R:  \expec{|\alg(R)|} \geq \frac 1c \cdot |\opt(R)|$.

\subparagraph*{The \maxpr\ Problem.} The problem of maximum throughput scheduling of packet requests on directed line
(\maxpr) is defined as follows. The input consists of: $n$ - the size of the
network, $B$ - node buffer capacities, $c$ - link capacities, and $M$ packet requests
$\{r_i\}_{i=1}^M$. The output is a schedule $S$. The goal is to maximize the
throughput of $S$.

\subsection{Path Packing in a uni-directed 2D-Grid}
\label{sec:prelim}
In this section we define a problem of maximum-cardinality path packing in a
two-dimensional uni-directed grid (\maxgrid). This problem is equivalent to  
\maxpr, and was used for that purpose
in previous work, where the formal reduction is also presented~\cite{AAF,ARSU,AZ,RR}.
For completeness, this reduction is given  in Appendix~\ref{sec:reduction}.
As the two problems  are
equivalent, we use  in the sequel terminology from both problems interchangeably.

The grid, denoted by $G^{st}=(V^{st},E^{st})$, is an infinite directed acyclic graph.
The vertex set $V^{st}$ equals $V\times\mathbb{N}$, where $V=\{0,1,\ldots, (n-1)\}$.
Note that we use the first coordinate (that corresponds to vertices in $V$) for the
$y$-axis and the second coordinate (that corresponds to time steps) for the $x$-axis (See Figure~\ref{fig:st} in Appendix~\ref{sec:reduction}).
The edge set consists of horizontal edges (also called store edges) directed to the
right and vertical edges (also called forward edges) directed upwards.  The capacity
of vertical edges is $c$ and the capacity of horizontal edges is $B$. We often refer
to $G^{st}$ as the space-time grid (in short, grid) because the $x$-axis is related to time and the
$y$-axis corresponds to the vertices in $V$.

A \emph{path request} in the grid is a tuple $r^{st}=(a_i,t_i,b_i)$, where $a_i,b_i\in V$
and $t_i\in \mathbb{N}$. The request is for a path that starts in node $(a_i,t_i)$
and ends in any node in the row of $b_i$ (i.e., the end of the path can be any node
$(b_i,t)$, where $t\geq t_i$).

A \emph{packing} is a set of paths $S^{st}$ that abides by the capacity constraints: For
every grid edge $e$, the number of paths in $S^{st}$ that contain $e$ is not greater
than the capacity of $e$.

Given a set of path requests $R^{st}=\{r^{st}_i\}_{i=1}^M$, the goal in the \maxgrid\
problem is to find a packing $S^{st}$ with the largest cardinality. (Each path in
$S^{st}$ serves a distinct path request.)

\subparagraph*{Multi-Commodity Flows (MCFs).}
Our use of path packing problems
gives rise to \emph{fractional} relaxations of that problem, namely to multi-commodity flows (MCFs) with unit
demands on uni-directional grids. The definitions and terminology of MCF's appear in Appendix~\ref{append:MCF}.

\subsection{Tiling, Classification, and Sketch Graphs}
\label{sec:def_tilings}

To define our algorithm we make use of partitions of the space-time grid described above
 into sub-grids. We define here the notions we use for this purpose.

\subparagraph*{Tiling.} \emph{Tiling} is a partitioning of the two-dimensional space-time
grid (in short, grid) into squares, called \emph{tiles}. Two parameters specify a
tiling: the side length $k$, an even integer,  of the squares, and the shifting $(\varphi_x,\varphi_y)$
of the squares.  The shifting refers to the $x$- and $y$-coordinates of the bottom
left corner of the tiles modulo $k$.  Thus, the tile $T_{i,j}$ is the subset of the
grid vertices defined by
\begin{align*}
  T_{i,j} &\triangleq \{ (v,t) \in V\times \mathbb{N} \mid ik  \leq
  v-\varphi_y < (i+1)k \text{ and } jk  \leq t -\varphi_x< (j+1)k\},
\end{align*}
where $\varphi_x$ and $\varphi_y$ denote the horizontal and vertical shifting, respectively.
We consider two possible shifts for each axis, namely, $\varphi_x,\varphi_y \in \{0,k/2\}$.

\subparagraph*{Quadrants and Classification.} Consider a tile $T$. Let $(x',y')$ denote the
lower left corner (i.e., south-west corner) of $T$. The \emph{south-west quadrant} of
$T$ is the set of vertices $(x,y)$ such that $x'\leq x < x'+k/2$ and $y'\leq y
< y'+k/2$.

For every vertex $(x,y)$ in the grid, there exists exactly one shifting
$(\varphi_x,\varphi_y)\in \{0,k/2\}^2$ such that $(x,y)$ falls in the south-west (SW)
quadrant of a tile.
Fix the tile side length $k$. We define a \emph{class} of requests for every shifting
$(\varphi_x,\varphi_y)$.  The class that corresponds to the shifting
$(\varphi_x,\varphi_y)$ consists of all the path requests $r^{st}_i$ whose origin
$(a_i,t_i)$ belongs to a SW quadrant of a tile in the tiling that uses the shifting
$(\varphi_x,\varphi_y)$.

\subparagraph*{Sketch graph and paths.} Consider a fixed tiling.  The \emph{sketch graph}
is the graph obtained from the grid after coalescing each tile into a single node.
There is a directed edge $(s_1,s_2)$ between two tiles $s_1,s_2$ in the sketch graph
if there is a directed edge $(\alpha,\beta)\in E^{st}$ such that $\alpha\in s_1$ and
$\beta\in s_2$.
Let $p^s$ denote the projection of a path $p$ in the grid to the sketch graph.  We
refer to $p^s$ as the \emph{sketch path} corresponding to $p$. Note that the length
of $p^s$ is at most $\ceil{|p|/k}+1$.

\section{Outline of our Algorithm}

Packet requests are categorized into four categories: very short, short, medium, and long, according to the source-destination distance of each packet. A
separate approximation algorithm is executed for each category.  The algorithm
returns a highest throughput solution among the solutions computed for the four
categories.

\subparagraph*{Notation.}
Three thresholds are used for defining very short, short, medium, and long requests:
$\medium  \triangleq 3\ln n, \short \triangleq 3\cdot \ln (\medium) = 3\cdot \ln (3\ln n)$, $\vshort \triangleq  3\cdot \ln (\short)= 3 \cdot \ln(3\cdot \ln (3\ln n))$.
\begin{definition}
A request $r_i$ is a \emph{very short} request if $b_i-a_i \leq \vshort$.
  A request $r_i$ is a \emph{short} request if $\vshort<b_i-a_i \leq \short$.  A request
  $r_i$ is a \emph{medium} request if $\short <b_i-a_i \leq \medium$.  A request
  $r_i$ is a \emph{long} request if $\medium < b_i-a_i $.
\end{definition}

We use a deterministic algorithm for the class of very short packets, and
in Lemma~\ref{lemma:vshort} we prove that this deterministic algorithm
achieves a constant approximation ratio. We use a randomized algorithm for
each of the classes of short, medium and long packets;  in Theorem~\ref{thm:final_randomized} we prove that this randomized algorithm achieves
a constant approximation ratio in expectation for each of these classes when $B=\Theta(c)$. Thus, we obtain the following:
\begin{theorem}[Main Result]
\label{th:main_result}
  There exists a randomized approximation algorithm for the \maxpr\
  problem that, when  $B=\Theta(c)$, achieves a constant approximation ratio in expectation.
\end{theorem}

\section{Approximation Algorithm for Very Short Packets}
In this section we present a constant ratio deterministic approximation algorithm for
very short packets for the case of $B=\Theta(c)$.  This algorithm, which is key to achieving the results of the present paper,
makes use  of a new combinatorial lemma that we prove in the next subsection, stating,
   roughly speaking,
   that if packets from a given set of packets
   are allowed to stay put in buffers  (i.e., use horizontal edges in the grid)
   only a limited number of time
   steps, $2d$ (where $d$ is the longest source-destination distance in the set of packets),
   then the
   optimal solution is decreased by only a bounded factor.
   We believe that this lemma may find additional applications in future work on routing and scheduling problems.

\subsection{Bounding Path Lengths in the Grid}
In this section we prove that bounding, from above,  the number of horizontal
edges along a path incurs only a small reduction in the throughput.
Previously known  bounds along these lines hold only for fractional solutions~\cite{AZ},
while we  prove here such claim for integral schedules.
Another similar variant is given  by Kleinberg and Tardos~\cite{KT}, where virtual circuit routing over \emph{undirected} grids is studied. It is proven in~\cite{KT} that a restricted integral optimal routing is a constant approximation to the unrestricted one:  routing requests have  origins and destinations within a subgrid of $2d$ by $2d$, and the restriction on the paths is that the routing is done within its supergrid of $4d$ by $4d$. As we will see below our lemma  holds w.r.t.~\emph{directed} graphs, that is, the lemma by~\cite{KT} does not handle our case.

Let $R_d$ denote a set of packet requests $r_i$, $i \geq 1$,  such that $b_i-a_i \leq d$ for any $i$.  Consider
the paths in the space-time grid that are allocated to the accepted requests in an optimal solution. We
prove that restricting the path lengths to $2d$ decreases both the optimal fractional
and  the optimal  \emph{integral} throughput  only by  a multiplicative factor of
$O(c/B)$. We note that if $B=\Theta(c)$, then we are guaranteed an optimal solution
which is only a 
constant fraction away from the unrestricted optimal solution.

\subparagraph*{Notation.}
For a single commodity acyclic flow $f_i$, let $\pmax(f_i)$ denote the diameter of
the support of $f_i$ (i.e., length of longest path\footnote{Without loss of
  generality, we may assume that each single commodity flow $f_i$ is acyclic.}).  For
  an MCF $F=\{f_i\}_{i\in I}$,  where $I$ is the set of flows, let $\pmax(F) \triangleq \max_{i\in I} \pmax(f_i)$.
  Let $\optfrac(R)$ (respectively, $\optint(R)$) denote a maximum throughput
  fractional (resp., integral) MCF with respect to the set of requests $R$.
  Similarly, let $\optfrac(R\mid \pmax < d')$ (respectively, $\optint(R\mid
  \pmax<d')$) denote a maximum throughput fractional (resp., integral) MCF with
  respect to the set of requests $R$ subject to the additional constraint that the
  maximum path length is at most $d'$. See Appendix~\ref{append:MCF} for further MCF terminology.

\begin{lemma}\label{lemma:bounded path length}
  \begin{align*}
    \optfrac(R_d \mid \pmax \leq 2d) &\geq \frac {c}{B+2c} \cdot \optfrac(R_d), \\
    \optint(R_d \mid \pmax \leq 2d) &\geq \frac {c}{2(B+c)} \cdot \optint(R_d).
  \end{align*}
\end{lemma}
\begin{proof}
  Partition the space-time grid into \emph{slabs} $S_j$ of ``width'' $d$. Slab $S_j$
  contains the vertices $(v,k)$, where $k\in [(j-1)\cdot d,j\cdot d]$, $j \geq 1$. We refer to
  vertices of the form $(v,jd)$ as the \emph{boundary} of $S_j$.  Note that if
  $v-u\leq d$, then the forward-only vertical path from $(u,jd)$ to $(v,jd+(v-u))$ is
  contained in slab $S_{j+1}$.

    We begin with the fractional case. Let $f^*=\optfrac(R_d)$ denote an optimal
  fractional solution for $R_d$.  Consider request $r_i$ and the corresponding single
  commodity flow $f^*_i$ in $f^*$. Decompose $f^*_i$ to flow-paths $\{p_\ell\}_\ell$.
  For each flow-path $p_\ell$ in $f^*_i$, let $p'_\ell$ denote the prefix of
  $p_\ell$ till it reaches the boundary of a slab. Note that $p'_\ell=p_\ell$ if
  $p_\ell$ is confined to a single slab. If $p'_\ell \subsetneq p_\ell$, then let
  $(v,jd)$ denote the last vertex of $p'_\ell$. Namely, the path $p'_\ell$ begins in
  $(a_i,t_i)\in S_j$ and ends in $(v,jd)$. Let $q''_\ell$ denote the forward-only path
  from $(v,jd)$ to $(b_i,jd+(b_i-v))$. (If $p'_\ell=p_\ell$, then $q''_\ell$ is
  an empty path.)  Note that $q''_\ell$ is confined to the slab $S_{j+1}$.  We refer
  to the vertex $(v,jd)$ in the intersection of $p'_\ell$ and $q''_\ell$ as the
  \emph{boundary} vertex.  Let $g_i$ denote the fractional single commodity flow
  for request $r_i$ obtained by adding the concatenated flow-paths $q_\ell\triangleq
  p'_\ell\circ q''_\ell$ each with the flow amount of $f^*_i$ along $p_\ell$.  Define
  the MCF $g$ by $g(e)\triangleq \sum_{i\in I} g_i(e)$. For every edge $e$, part of
  the flow $g(e)$ is due to prefixes $p'_\ell$, and the remaining flow is due to
  suffixes $q''_\ell$.  We denote the part due to prefixes by $g_{pre}(e)$ and refer
  to it as the prefix-flow. We denote the part due to suffixes by $g_{suf}(e)$ and
  refer to it as the suffix-flow. By definition, $g(e)=g_{pre}(e)+g_{suf}(e)$.

  The support of $g_i$ is contained in the union of two consecutive slabs. Hence, the
  diameter of the support of $g_i$ is bounded by $2d$. Hence $\pmax(g)\leq 2d$.

  Clearly, $|g_i|=|f^*_i|$ and hence $|g|=|f^*|$.  Set $\rho=c/(B+2c)$. To complete
  the proof, it suffices to prove that $\rho \cdot g$ satisfies the capacity
  constraints.  Indeed, for a ``store'' edge $e=(v,t)\rightarrow (v,t+1)$, we have
  $g_{suf}(e)=0$ and $g_{pre}(e)\leq f^*(e)\leq B$. For a ``forward'' edge
  $e=(v,t)\rightarrow (v+1,t+1)$ we have $g_{pre}(e) \leq f^*(e)\leq c$.  On the
  other hand, $g_{suf}(e)\leq B+c$.  The reason is as follows.  All the suffix-flows
  along $e$ start in the same boundary vertex $(u,jd)$ below $e$.  The amount
  of flow forwarded by $(u,jd)$ is bounded from above by the amount of incoming flow, which is
  bounded by $B+c$. This completes the proof of the fractional case.

  We now prove the integral case. The proof is a variation of the proof for the
  fractional case in which the supports of pre-flows and suffix-flows are disjoint.
  Namely, one alternates between slabs that support prefix-flows and slabs that
  support suffix-flows.

  In the integral case, each accepted request $r_i$ is allocated a single path $p_i$,
  and the allocated paths satisfy the capacity constraints. As in the fractional
  case, let $q_i\triangleq p'_i\circ q''_i$, where $p'_i$ is the prefix of $p_i$ till
  a boundary vertex $(v,jd)$, and $q''_i$ is a forward-only path.  We need to prove
  that there exists a subset of at least $c/(2(B+c))$ of the paths $\{q_i\}_i$ that
  satisfy the capacity constraints. This subset is constructed in two steps.

  First, partition the requests into ``even'' and ``odd'' requests according to the
  parity of the slab that contains their origin $(a_i,t_i)$.  (The parity of request
  $r_i$ is simply the parity of $\ceil{t_i/d}$.)  Pick a part that has at least half
  of the accepted requests in $\optint(R_d)$; assume w.l.o.g. that such a part is the
  part of the even slabs. Then, we only keep accepted
  requests whose origin belong to even slabs.

  In the second step, we consider all boundary vertices $(v,j\cdot d)$. For each
  boundary vertex, we keep up to $c$ paths that traverse it, and delete the
  remaining paths if such paths exist.  In the second step, again, at least a $c/(B+c)$ fraction
   of
  the paths survive.  It follows that altogether at least $c/(2(B+c))$ of the paths
  survive.

  We claim that the remaining paths satisfy the capacity constraints.  Note that
  prefixes are restricted to even slabs, and suffixes are restricted to odd slabs.
  Thus, intersections, if any, are between two prefixes or two suffixes.  Prefixes
  satisfy the capacity constraints because they are prefixes of $\optint(R_d)$.
  Suffixes satisfy the capacity constraints because if two suffixes intersect, then
  they start in the same boundary vertex.  However, at most $c$ paths emanating from
  every boundary vertex survive.  Hence, the surviving paths satisfy the capacity
  constraints, as required. This completes the proof of the lemma.
\end{proof}

We note that if $B=\Theta(c)$, then Lemma~\ref{lemma:bounded path length}
guarantees a restricted  optimal solution which is only a  constant fraction away from the unrestricted optimal solution.

\subsection{The Algorithm for Very Short Packets}
 In this section we present a deterministic approximation algorithm for
very short packets, whose approximation ratio is constant when $B=\Theta(c)$.

The very short requests are  partitioned into four classes, defined as follows. Consider
four tilings each with side length $k\triangleq 4\vshort$  and horizontal and vertical
shifts in $\varphi_x,\varphi_y\in \{0,k/2\}$.
The four possible shifts define four
classes: The packets of a certain class (shift) are the packets whose source nodes reside
in the SW quadrants of the tiles according to a given shift.  Observe that each packet
request belongs to exactly one class. We
say that a path $p_i$ from $(a_i,t_i)$ to the row of $b_i$ is \emph{confined to a
  tile} if $p_i$ is contained in one tile.

It is now possible to efficiently compute a constant approximation throughput
solution for each class, under the restriction that each path is of length at most $2\vshort$. Note
that this restriction means that those paths are confined to the tile that contains the origin of the path, thus this  solution
can be computed for each tile separately. On the other hand, by
Lemma~\ref{lemma:bounded path length},
 this restriction reduces
 the cardinality of
the optimal solution compared to the unrestricted optimal solution for that class by only a constant factor,
when $B=\Theta(c)$.
The algorithm computes a constant approximation (bounded path length) solution for each class, and
returns a highest throughput solution among the four solutions.

The polynomial deterministic brute force algorithm that we use is essentially the same as the one used in~\cite{RR} for a similar situation.
For completeness, below we
state it and prove its polynomial running time.

\begin{lemma}{\cite[Lemma 7]{RR}}\label{lemma:vshort}
Consider a tile $T$ of dimensions $k\times k$, for $k=4 \vshort$.
Given a set of requests $(a_i,b_i,t_i)$ all with source node in the SW quadrant of T and such that $b_i - a_i \leq \vshort$,
 consider the optimal solution for this set when all paths are restricted to length at most $2\vshort$.
There is a polynomial-time deterministic algorithm that finds a constant approximation to that optimal solution.
\end{lemma}

\begin{proof}
   The constant factor approximation algorithm for a tile uses the following
  brute-force approach.
  First observe that since  $b_i - a_i \leq \vshort$ and the paths are restricted to length at most $2\vshort$, then all paths of
  the optimal solution in question are confined to $T$.
  Define $P$ to be the set of all  paths connecting two end-points in $T$.
  Since $T$ is of size $k \times k$  there are
  at most
  $|P|=O(k^2\cdot 2k \cdot 2^{2k})$ paths
connecting end-points in $T$ (this is an overestimation). Define a candidate solution to be a choice of
the number of messages along every path in $P$; note that each of these numbers is bounded from above by
$\min\{c,M\}$. Observe that
one can in polynomial time check if a candidate solution is feasible.\footnote{There are two checks
to be done: (1) whether the
candidate solution does not violate capacities; and (2) whether the candidate solution
is coherent with the set of input packets.
The first check can be done by going over all edges in the tile, and for each such edge summing the numbers
associated with all the
paths that go through that edge, checking that this sum does not exceed the capacity of that edge. The second check can
be done by going over the $|P|$ paths: each such path serves a well defined request since it departs from a given
grid-node $(a,t)$, and
reaches a row $b$, thus it serves a request $(a,b,t)$. For each path serving requests $(a,b,t)$, check that
 the number associated with that path does not exceed the number of requests $(a,b,t)$ in the input.}

The basic idea is that the brute-force algorithm generates all candidate
solutions, checks them for feasibility, and chooses the best one among the feasible ones.
However, this may still result in a too large number of candidate solutions to check. Therefore, the algorithm only generates
candidate solutions where for all paths the number of messages along that path is a power of $2$, or $0$.
This only decreases the value of the  best candidate solution by a
factor of at most $2$. Thus, the number of candidate solutions checked is at most $O((\log M)^{|P|})$.
Since $|P|=O(k^2\cdot 2k \cdot 2^{2k})$ and $k=4\vshort=O(\log \log \log  n)$ the number of candidate solutions
checked is polynomial
in $n$ and $M$.\footnote{Let $s=M+n$. The number of candidate solutions checked is
$O((\log M)^{|P|}) = O((\log s)^{2^{O(\log \log \log s)}})=O(2^{(\log \log s)^{O(1)}})=o(s^\alpha)$, for any constant $\alpha$.}
\end{proof}

\section{Approximation Algorithm for Short, Medium and Long Requests}
\label{sec:randomized_algo}

In this section we give a randomized algorithm that will be used for the three classes of short, medium, and long requests.
Let $\gamma=\min\{B,c\}$.
When run on a given class (among short, medium, long requests), the algorithm given in this section produces an integral solution
with expected cardinality  at most a constant
factor away from the optimum {\em fractional} solution (for the same class) on a network with both edge  and buffer
capacities equal to  $\gamma$. Observe that when moving from the original network with capacities $B$ and $c$ to
a network with capacities
$\gamma$, the {\em fractional} optimum looses a factor of at most $\max\{B/\gamma,c/\gamma\}=\max\{B/c,c/B\}$. Thus, the algorithm of this section is
an $O(\max\{B/c,c/B\})$-approximation algorithm with respect to
 the fractional   optimum solution for each class, and hence also with respect to the integral optimum of  each class.
When $B=\Theta(c)$ we thus get a randomized algorithm with expected constant approximation for each of the three classes
treated in this section.

\subparagraph*{Notation.}
Let $R_{\dmin,\dmax}$ denote the set of packet requests whose
source-to-destination distance is greater than $\dmin $ and at most $\dmax$.
Formally, $
  R_{\dmin,\dmax}  \triangleq \{ r_i \mid \dmin < b_i-a_i\leq \dmax\}.
$
\subparagraph*{Parametrization.}
When applied to medium requests we use the parameter $\dmax=\medium$ and
$\dmin=\short$. When applied to long requests the parameters are $\dmax=n$ and
$\dmin=\medium$. Note that these parameters satisfy $\dmin = 3\cdot \ln \dmax$.

\subparagraph*{Chernoff Bound.}
We use the following Chernoff bound in the analysis of the algorithm for short, medium and long requests.
\begin{definition}\label{def:beta}
  The function $\beta:(-1,\infty) \rightarrow \mathbb{R}$ is defined by $\beta(\eps)\eqdf (1+\eps)\ln
  (1+\eps) - \eps$.
\end{definition}

\begin{theorem}[Chernoff Bound~\cite{raghavan1986randomized,young1995randomized}]
\label{thm:Chernoff}
Let $\{X_i\}_i$ denote a sequence of independent random variables attaining values in $[0,1]$.
Assume that $\expec{X_i}\leq \mu_i$. Let $X\eqdf \sum_i X_i$ and $\mu\eqdf\sum_i \mu_i$.
Then, for $\eps >0$,
\begin{align*}
  \prbig{X\geq (1+\eps)\cdot \mu} &\leq e^{-\beta(\eps)\cdot \mu}.
\end{align*}
\end{theorem}
\begin{corollary}\label{coro:Chernoff}
  Under the same conditions as in Theorem~\ref{thm:Chernoff},
\begin{align*}
  \prbig{X\geq \alpha\cdot \mu} &\leq \left(\frac{e}{\alpha}\right)^{\alpha\cdot \mu}.
\end{align*}
\end{corollary}

\subsection{The Algorithm for $R_{\dmin,\dmax}$}
\label{sec:algo_def}
The algorithm for $R_{\dmin,\dmax}$ proceeds as follows.
To simplify notation, we abbreviate $R_{\dmin,\dmax}$ by $R$.
The parameters $\dmin$ and $\dmax$ must satisfy that
$\dmin  =  3\cdot \ln \dmax$.
We use the randomized rounding procedure by Raghavan and Thompson~\cite{raghavan1986randomized,raghavan1987randomized}. The description of this randomized rounding procedure is deferred to Appendix~\ref{sec:RR}.
To run the following algorithms we  reduce the packet requests in $R$ to path requests $R^{st}$ in a space-time graph $\Gst$ with edge capacities $\gamma$. The following algorithms is to give a solution to the problem of routing a maximum cardinality
subset of $R^{st}$ on the graph $\Gst$.

\begin{enumerate}
\item
\label{step:old2}
On $\Gst$,
 compute a maximum throughput fractional MCF $F\triangleq\{f_i\}_{r_i\in
    R^{st}}$ with edge capacities $\tilde{c}(e)=\lambda \cdot \gamma$,  for $\lambda=1/(\beta(3) \cdot 6)$, and bounded diameter $\pmax(F)\leq 2\dmax$.  We remark that this MCF can be
  computed in time polynomial in $n$ - the number of nodes, and $M$ - the number of
  requests.\footnote{Note that the requests in $R^{st}$, as defined in Section~\ref{sec:prelim}, are from a grid node to a grid {\em row}. To be fully coherent with standard MCF terminology and notations, one would
  need to add for each row in the grid a super-node, connect all nodes on that row to this super-node with edges of capacity say, $M$, and define the MCF problem with flow requests from grid nodes to these super nodes instead of the corresponding rows. Further note that since always $\dmax\leq n$, we can consider
   a space-time grid of size at most $ n \times (M \cdot 2n)$, which can be constructed by
   going over the release times of all $M$ requests, eliminating ``unnecessary'' time steps.
  One can then compute a maximum throughput fractional solution with bounded diameter on this
   grid using linear programming.  This is true because the constraint $\pmax(f_i)\leq d'$
  is a linear constraint and can be imposed by a polynomial number of inequalities
  (i.e, polynomial in $n$ and $d'$). For example, one can construct a product network
  with $(d'+1)$ layers, and solve the MCF problem over this product graph.}
\item
\label{step:old3}
Partition $R^{st}$ into $4$ classes $\{R^{j}\}_{j=1}^4$ according to the shift of tiling that results in the source
node being in the SW quadrant of a  $k\times k$ tiling, where $k\triangleq  2\dmin = 6\ln \dmax$
 (see Section~\ref{sec:def_tilings}).  Pick a class $R^j$ such that the throughput of $F$ restricted to $R^j$
  is at least a quarter of the throughput of $F$, i.e., $|F(R^j)|\geq |F|/4$.
\item
\label{step:old4}
For each request $r_i\in R^j$, apply randomized rounding independently to
  $f_i$ (as described in Appendix~\ref{sec:RR}).
  The outcome of randomized rounding
  per request $r_i\in R^j$ is either ``reject'' or a path $p_i$ in $\Gst$. Let
  $\Rrnd \subseteq R^j$ denote the subset of requests $r_i$ that are assigned a
  path $p_i$ by the randomized rounding procedure.
\item
\label{step:old5}
Let $\Rfilter\subseteq \Rrnd$ denote the requests that remain after
  applying filtering (described in Section~\ref{sec:filtering}).
\item
\label{step:old6}
Let $\Rquad\subseteq \Rfilter$ denote the requests for which routing
  in first quadrant is successful (as described in Section~\ref{sec:quadrant}).
\item
\label{step:old7}
Complete the path of each request in $\Rquad$ by applying crossbar routing (as
  described in Section~\ref{sec:detailed}).
\end{enumerate}

\subsection{Filtering}\label{sec:filtering}
\subparagraph*{Notation.}
Let $e$ denote an edge in the space-time grid $G^{st}$.  Let $e^s$ denote an edge in
the sketch graph (see Section~\ref{sec:def_tilings}).
We view $e^s$ also as the set of edges in
$G^{st}$ that cross the tile boundary that corresponds to the sketch graph edge $e^s$.
The path $p_i$ is a random variable that denotes the path, if any, that is chosen for
request $r_i$ by the randomized rounding procedure.
For a path $p$ and an edge $e$ let $\ind{p}(e)$ denote the $0$-$1$ indicator function
that equals $1$ iff $e\in p$.

\medskip\noindent
The set of filtered requests $\Rfilter$ is defined as follows (recall that $\lambda=1/(\beta(3)\cdot 6)$).
\begin{definition}
  A request $r_i$ is $r_i\in \Rfilter$ if and only if $r_i$ is accepted by the randomized rounding
  procedure, and for every sketch-edge $e^s$ in the sketch-path $p^s_i$ (see Section~\ref{sec:def_tilings}) it holds that
      $$\sum_{j :r_j \in \Rrnd} \ind{p^s_j}(e^s) \leq 4\lambda\cdot k \cdot \gamma \:.$$
      \end{definition}

We now give a lower bound on the cardinality of the set of requests that pass the filtering stage.
\begin{claim} \label{claim:filter}
Let  $k=6\ln \dmax$.
$\expec{|\Rfilter|}\geq \left(1-O(\frac{1}{k})\right)\cdot \expec{|\Rrnd|}$.
\end{claim}

\begin{proof} We begin by bounding  from above the probability that more than $4\lambda k \gamma$ sketch
  paths cross a given sketch edge.

\begin{lemma}\label{lemma:chernoffrr}
For every edge $e^s$ in the sketch graph,\footnote{The
    $e$ in the RHS is the base of the natural logarithm.}
  \begin{align}
    \label{eqn:chernoff}
    \prbig{\sum_i \ind{p^s_i}(e^s) > 4\lambda k \gamma} \leq e^{-k/6}\:.
  \end{align}
\end{lemma}
\begin{proof}[Proof of lemma]
  Since the demand of each request is $1$, it follows that $f_i(e^s) \leq 1$, for any request $r_i$ and any sketch graph edge
  $e^s$.  Thus, for every edge $e^s$ and request $r_i$, we have
  $\expec{\ind{p^s_i}(e^s)} = \prbig{\ind{p^s_i}(e^s) = 1} =f_i(e^s) \leq 1$.
Fix a sketch edge $e^s$. The random variables $\{\ind{p^s_i}(e^s)\}_{i}$ are
independent $0$-$1$ variables.  Moreover, $\sum_i \expec{\ind{p^s_i}(e^s)} =\sum_i
f_i(e^s) =\sum_{e\in e^s} \sum_i f_i(e)\leq k \cdot \lambda \gamma$.
 By Chernoff bound~\footnote{We use the following version of Chernoff
 Bound~\cite{raghavan1986randomized,young1995randomized}.
Let $\{X_i\}_i$ denote a sequence of independent random variables attaining values in $[0,1]$.
Assume that $\expec{X_i}\leq \mu_i$. Let $X\eqdf \sum_i X_i$ and $\mu\eqdf\sum_i \mu_i$.
Then, for $\eps >0$,
$
  \prbig{X\geq (1+\eps)\cdot \mu} \leq e^{-\beta(\eps)\cdot \mu}.
$}
\begin{align*}
  \prbig{\sum_i \ind{p^s_i}(e^s) > 4\cdot \sum_i \expec{\ind{p^s_i}(e^s)}} &<
  e^{-\beta(3)\cdot k\lambda \gamma} \leq e^{-k/6}~,
\end{align*}
since $\gamma \geq 1$.
\end{proof}

  A request $r_i\in \Rrnd$ is not in $\Rfilter$ iff at least one of the edges $e^s\in p^s_i$
 has more than $4k \lambda  \gamma$ paths on it.
  Hence, by a union bound,
  \begin{align*}
    \prbig{r_i \not\in \Rfilter \mid r_i \in \Rrnd} & \leq |p^s_i| \cdot e^{-k/6} \leq \left(\ceil*{\frac{2 \dmax}{k}}+2\right) \cdot e^{-\ln \dmax}
= O\left(\frac 1k\right),
  \end{align*}
  since $k=6\ln \dmax$.
  \end{proof}

\subsection{Routing in the First Quadrant}\label{sec:quadrant}
In this section, we deal with the issue of evicting as many requests as possible
from their origin quadrant to the boundary of their origin quadrant.
\begin{remark}
  Because $k/2\leq \dmin$ every request that starts in a SW quadrant of a tile
  must reach the boundary (i.e., the extreme nodes on the top or right side) of the quadrant before it  reaches its destination.
\end{remark}

\subparagraph*{The maximum flow algorithm.}
Consider a tile $T$.  Let $X$ denote A set of requests $r_i$ whose
source $(a_i,t_i)$ is in the south-west quadrant of $T$.  We say that a subset
$X'\subseteq X$ is \emph{quadrant feasible} (in short, feasible) if it satisfies
the following condition: There exists a set of paths, creating a load of at most $\gamma$ on each edge,  $\{q_i \mid r_i\in
X'\}$, where each path $q_i$ starts in the source $(a_i,t_i)$ of $r_i$ and ends in
the top or right side of the SW quadrant of $T$.

\medskip\noindent
We employ a maximum-flow algorithm to solve the following  problem.
\begin{description}
\item[Input:] A set of requests $X$ whose source is in the SW
  quadrant of $T$.
\item[Goal:] Compute a maximum cardinality quadrant-feasible subset $X'\subseteq X$. In addition, for each request $r \in X'$,
compute a path from the source node of $r$ to a  node on  the boundary of the SW quadrant of $T$.
\end{description}

The algorithm is simply a maximum-flow algorithm over the following network, denoted
by $N(X)$. Augment the quadrant with a super source $\tilde{s}$ and a super sink
$\tilde{t}$.
The super source $\tilde{s}$ is
connected to every source $(a_i,t_i)$ of a request $r_i\in X$ with a unit capacity
directed edge. (If $\alpha$ requests share the same source, then the capacity of the
edge is $\alpha$.)  There is a $\gamma$-capacity edge from every vertex in the top side and right side
of the SW quadrant of $T$ to the super sink $\tilde{t}$.  All the grid edges are assigned $\gamma$
capacities.  Compute an integral maximum flow in the network. Decompose the flow to
unit-flow paths. These flow paths are the paths that are allocated to the requests in
$X'$.

\subparagraph*{Analysis.}
Fix a tile $T$ and let $R_T\subseteq \Rfilter$ denote the set of requests in
$\Rfilter$ whose source vertex is in the SW quadrant of $T$. Let $R'_T\subseteq R_T$
denote the maximum cardinality quadrant-feasible subset of $R_T$ as computed by the max-flow
algorithm above.
Let $\Rquad = \bigcup_T R'_T$.

\medskip\noindent We now prove the following theorem that relates the {\em expected value} of  $|R'_T|$ to
 the expected value of $|R_T|$.  Observe that it is not always true that the same relation holds for any specific $R_T$ that results from a specific  realization of randomized rounding procedure.

 \begin{theorem}{\cite{kleinberg1996approximation,RR}}
\label{thm:quadrant}
 $\expectau{ |\Rquad| } \geq 0.93\cdot \expectau{|\Rfilter|}$, where
  $\tau$  is the probability space induced by the randomized rounding procedure.
\end{theorem}
\begin{proof}
  By linearity of expectation, it suffices to prove that $\expectau{ |R'_T| } \geq
  0.93\cdot \expectau{|R_T|}$, for any given tile $T$.

The proof below will go along the following lines.
 We define a certain capacity constraint
 over rectangles in the tile; this definition makes use of the capacity of the boundary of the rectangles, and the
  number of requests having their origin within them. We define the set $\hat{R}_{T} \subseteq R_T$
  to be a set of requests based on the capacity
  constraints of the rectangles containing the origin of the requests.
  We prove that:
  (1)~The set $\hat{R}_{T}$ thus  defined is quadrant-feasible, and
  (2)~$\expectau{ |\hat{R}_T| } \geq
  0.93 \cdot \expectau{|R_T|}$.
  By the algorithm, $R'_T$ is a maximum
  cardinality (maximum flow) set, therefore, $|R'_T|\geq |\hat{R}_T|$, and the theorem follows.

  We now describe how the quadrant-feasible subset $\hat{R}_{T}$ is defined.

  Consider a
  subset $S$ of the vertices in the SW quadrant of $T$. Let $\dem_Y(S)$ denote the
  number of requests in $Y\subseteq R_T$ whose origin is in $S$. Let $\capa(S)$ denote the capacity of
  the edges, in the network $N(R_T)$, that emanate from $S$.  By the min-cut max-flow
  theorem, a set of requests $Y\subseteq R_T$ is quadrant-feasible if and only if
  $\dem_Y(S)\leq \capa(S)$ for every cut $S\cup\{\tilde{s}\}$ in the network $N(R_T)$.
  But, to establish this condition, it is not necessary to consider all the cuts. It  suffices to consider only
  axis parallel rectangles contained in the SW quadrant of $T$; a set of requests $Y\subseteq R_T$ is
  quadrant-feasible if and only if
  $\dem_Y(Z)\leq \capa(Z)$ for every  axis parallel rectangle $Z$
  contained in the SW quadrant of $T$.
   The reason is as follows. Without
  loss of generality the set $S$ is connected in the underlying undirected graph of
  the grid (i.e., consider each connected component of $S$ separately; if  the condition does not hold for
  $S$, then it does not hold for at least one
  of its connected components). Every
  connected set $S$ can be replaced by the smallest rectangle $Z(S)$ that
  contains $S$.  We claim that $\capa(S)\geq \capa(Z(S))$ and $\dem_Y(S)\leq
  \dem_Y(Z(S))$. Indeed, there is an injection from the edges in the cut of $Z(S)$ to
  the edges in the cut of $S$. For example, a vertical edge $e$ in the cut of $Z(S)$
  is mapped to the topmost edge $e'$ in the cut of $S$ that is in the column of $e$.
  Hence, $\capa(Z(S))\leq \capa(S)$. On the other hand, as $S\subseteq Z(S)$, it
  follows that $\dem_Y(S)\leq \dem_Y(Z(S))$. Hence if $\dem_Y(S)>\capa(S)$, then
  $\dem_Y(Z(S))> \capa (Z(S))$.

  We say that a rectangle $Z$ is \emph{overloaded} with respect to a set of requests $Y$  if $\dem_Y(Z)>\capa(Z)$. The set
  $\hat{R}_T \subseteq R_T$ is defined to be the set of requests such that  $r_i\in \hat{R}_T$  iff the origin
  of  $r_i$ is not included in any overloaded (with respect to $R_T$) rectangle. Namely,
$$\hat{R}_T \triangleq \{r_i\in R_T  \mid \neg\exists ~Z ~\mbox{s.t.}~ (a_i,t_i) \in Z \text{ and } Z  \text{ is overloaded with respect to $R_T$} \}\:$$
Consider a rectangle $Z$ with dimensions $x\times y$.
We wish to bound from above the probability that
  $\dem_{R_T}(Z)>\capa(Z)=\gamma\cdot(x+y)$.
  Since requests  in $R_T$ with origin in $Z$ must  exit the quadrant,  and hence must exit $Z$,
  it follows that $\dem_{R_T}(Z)$ is
  bounded from above by the number of paths in $R_T$ that cross the top side or the right side of $Z$ (note that there might be  paths
  that cross these sides, but  do not start in $Z$).  The amount of flow that emanates from $Z$ is bounded by $\lambda\cdot \gamma\cdot (x+y)$
  (the  capacities for the flow algorithms are $\lambda \cdot \gamma $ and there are $x+y$ edges in the cut). By the randomized rounding procedure,  
   for every edge $e$ and every request $i$, $\prbig{e\in p_i} = f_i(e)$.  Summing  over all the
  edges that emanate  from $Z$ and all the requests in $R_T$, the expected number of
  paths (of requests in $R_T$) which emanate from
   $Z$ equals the total flow of the requests in $R_T$ that cross the top side or the right side of $Z$. This quantity,  in turn, is
  bounded from above by the capacity $\lambda\cdot \gamma \cdot (x+y)$. As the paths of the requests are
  independent random variables, we obtain:\footnote{Using the following version of the Chernoff bound:
   $\prbig{X\geq \alpha\cdot \mu} \leq \left(\frac{e}{\alpha}\right)^{\alpha\cdot \mu}$. The $e$ in the formulae
   denotes the basis of the natural logarithm, not an edge.}

 \begin{align*}
    \prbig{\dem_{R_T}(Z)> \capa(Z)}  \leq \prbig{\sum_{i\in{R_T}} |p_i \cap \cut(Z)| >
      \gamma \cdot (x+y)} \leq ( \lambda \cdot e)^{\gamma\cdot(x+y)} \leq ( \lambda \cdot e) ^{x+y}~,
\end{align*}
       since
      $\lambda \cdot e <  1$ and $\gamma \geq 1$.

  For each $x,y$, each source $(a_i,t_i)$ is contained in at most $x\cdot y$ rectangles with dimensions
   $x\times y$. By applying a union bound, the probability that $(a_i,t_i)$ is
  contained in an overloaded rectangle is bounded from above by
\begin{align}
  \prbig{\exists \text{ overloaded rectangle } Z : (a_i,t_i)\in Z} &\leq \sum_{x=1}^{\infty}
\sum_{y=1}^{\infty} xy \cdot ( \lambda \cdot e)^{x+y} \leq \frac{(\lambda \cdot e)^2}{(1-\lambda \cdot e)^4} \leq 0.07 \label{eq:rectangles},
\end{align}
and the theorem follows.
\end{proof}

Routing inside the tiles  (see Section~\ref{sec:detailed}) requires however a certain
upper bound on the number of requests that start in a tile and emanate from each side of the SW quadrant. Namely that for each side of the SW quadrant
at most $\gamma\cdot (k/3)$ paths  that start in that quadrant  reach that side of the quadrant.
Using a simple procedure, i.e., taking the solution produced by the maximum flow algorithm above
and greedily eliminating paths, we get a solution
for which this condition holds, and with cardinality only a constant factor smaller.
\begin{corollary}\label{coro:limit}
Let $\Rquad$ be the set of quadrant-feasible paths such that at most $\gamma\cdot (k/3)$ paths reach each
side of each quadrant. Then,
   $\expectau{ |\Rquad| } \geq \Omega(1)\cdot \expectau{|\Rfilter|}$, where
  $\tau$  is the probability space induced by the randomized rounding procedure.
\end{corollary}
\begin{proof}
  The sum of the capacities of the edges emanating from a side of the quadrant is
  $\gamma\cdot(k/2)$.  Limiting the number of paths to $\gamma\cdot(k/3)$ reduces the throughput by at most a
  factor of $2/3$.
\end{proof}

\subsection{Detailed Routing}\label{sec:detailed}

In this section we deal with computing paths for requests $r_i\in \Rquad$ starting
from the boundary of the SW quadrant that contains the source $(a_i,t_i)$ till the
destination row $b_i$. These paths are concatenated to the paths computed in the
first quadrant to obtain the {\em final} paths of the accepted requests. Detailed
routing is based on the following components: (1)~The projections of both the final path
and  of the path $p_i$ on the sketch graph must coincide.  (2)~Each tile is partitioned
to quadrants and routing rules within a tile are defined. (3)~Crossbar routing within
each quadrant is applied to determine the final paths (except for routing in SW
quadrants in which paths are already assigned).

\subparagraph*{Sketch paths and routing between tiles.}
Each path $p_i$ computed by the randomized rounding procedure is projected to a
sketch path $p^s_i$ in the sketch graph. The final path $\hat{p}_i$ assigned to
request $r_i$ traverses the same sequence of tiles, namely, the projection of
$\hat{p}_i$ is also $p^s_i$.

\subparagraph*{Routing rules within a tile~\cite{icalpEvenM10}.}
Each tile is partitioned to quadrants as depicted in Figure~\ref{fig:quads}. The bold
sides (i.e., ``walls'') of the quadrants indicate that final paths may not cross
these walls. The classification of the requests ensures that source vertices of
requests reside only in SW quadrants of tiles. Final paths may not enter the SW
quadrants; they may only emanate from them. If the endpoint of a sketch path $p^s_i$
ends in tile $T$, then the path $\hat{p}_i$ must reach a copy of its destination row
$b_i$ in $T$. Reaching the destination row $b_i$ is guaranteed by having $\hat{p}_i$ reach  the
top row of the NE quadrant of $T$ (and thus it must reach the row of $b_i$ along the
way).

\begin{figure}
    \centering
\begin{subfigure}{0.36\textwidth}
\centering
\includegraphics[width=\textwidth]{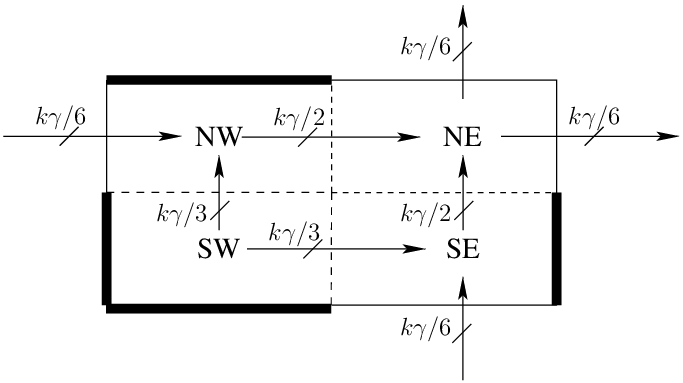}
\caption{}
\label{fig:quads}
\end{subfigure}
\quad \quad \quad \quad \quad \quad \quad
\begin{subfigure}{0.35\textwidth}
\centering
\includegraphics[width=\textwidth]{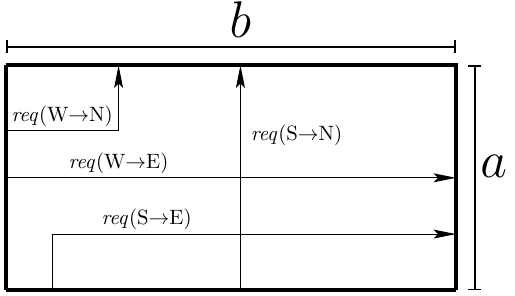}
\caption{}
\label{fig:crossbar}
\end{subfigure}
\caption{(a)~Partitioning of a tile to quadrants. Thick lines
represent ``walls'' that cannot be crosses by paths. Sources  of requests may reside
only in the SW quadrant of a tile. Maximum flow amounts crossing quadrant
sides appear next to each side. Final destinations of paths are assumed
(pessimistically) to be in the top row of the NE quadrant. (b)~4 types of requests within an $a\times b$ quadrant (in our case $a=b=k$).  
If the total flow destined to the north or east sides is at most their capacities then the detailed routing within the quadrant succeeds. 
 }
\end{figure}

\subparagraph*{Crossbar routing.~\cite{spaaEvenMP15}.} Routing in each quadrant is simply
an instance of routing in a (uni-directional) 2D grid where requests enter from two adjacent sides and
exit from the opposite sides, that is, there are 4 types of requests: for $X \in \{S, W\}$
and $Y \in \{N, E\}$, let $\rm{req}(X \rightarrow Y)$ denote the set of path requests whose entry point  to the tile is in the $X$ side and whose exit point is the $Y$ side. 
In fact, only the NE quadrant has all four types of requests cross it. 
The SE and NW quadrants have only two out of the four types. 
Figure~\ref{fig:crossbar} depicts such an
instance  of routing in one of the quadrants, in which requests arrive from the left and bottom sides and exit from the top and right sides.
To show that crossbar routing within a quadrant succeeds in our case, we use the following claim from~\cite{spaaEvenMP15}.
\begin{claim}{\cite[Proposition~5, Remark~6]{spaaEvenMP15}}\label{claim:crossbar}
  Consider a $2$-dimensional directed $a\times b$ grid with edges of uniform capacity. A set of requests can be
  routed from the bottom and left boundaries of the grid to the opposite boundaries,
  if and only if the number of requests that should exit each side is at most the total capacity of edges crossing that  side.
\end{claim}

\medskip\noindent
We conclude with the following claim.
\begin{claim}
  Detailed routing succeeds in routing all the requests in $\Rquad$ to create final paths for all requests in $\Rquad$.
\end{claim}
\begin{proof}
  The sketch graph is a directed acyclic graph.  Sort the tiles in topological
  order and within each tile, order the quadrants also in topological order: SW, NW,
  SE, NE, to get a topological order of all quadrants in the sketch graph. 
  We prove by induction on the position of the quadrant in that topological
  order that detailed routing up to and including that quadrant succeeds.
  The claim for
  all SW quadrants follows because  no path enters these quadrants from the outside and routing within these quadrants for requests with sources in these quadrants is identical to the paths computed by the randomized routing procedure. The SW quadrant of the first  tile (according to the topological order) establishes also the
  basis of the induction.
  We now note that filtering ensures that the number of paths that cross each tile boundary
  is at most $2\lambda k \gamma< k\gamma/6$,\footnote{This follows since $\beta(3)>2$, and since $\lambda = \frac{1}{\beta(3)\cdot 6}$.}
  and that the number of paths that cross each of the boundaries of each SW quadrant is at most
  $\gamma\cdot(k/3)$ (see Corollary~\ref{coro:limit}).
  Further note that for each request entering a tile on a certain boundary and having to exit  that tile on a certain other boundary, the sequence of quadrants that it has to cross within the tile is fixed. Therefore, the number of requests that enter each  quadrant on a certain quadrant-boundary  and the number of requests that have to exit  this quadrant through a certain other quadrant-boundary
is as depicted in Figure~\ref{fig:crossbar}.
  The induction
  step follows by applying Claim~\ref{claim:crossbar}.
\end{proof} 

\subsection{Approximation Ratio}

\begin{theorem}\label{thm:medium and long}

  The approximation ratio of the algorithm for packet requests in $R_{\dmin,\dmax}$,
  for $\dmin  =  3\cdot \ln \dmax$, on network of arbitrary capacity $\gamma$,
  is constant in expectation.
  \end{theorem}
  \begin{proof}

  We follow the algorithm, as defined in Section~\ref{sec:algo_def},  stage by stage.

    Stage~\ref{step:old2} computes a fractional maximum multi-commodity flow on a network with  edge capacities $\lambda \cdot \gamma$
   and with the requirement that all flows have bounded diameter of $2\dmax$.
    By Lemma~\ref{lemma:bounded path length}, bounding path lengths in the MCF
    results in a solution of at least a $1/3$ fraction of the unrestricted one, and
    the scaling of the capacities in the space-time grid from $\gamma$ to $\lambda \cdot \gamma$
    results in a solution which is at least a
    $\lambda=\Omega(1)$ fraction of the latter.

    Stage~\ref{step:old3} classifies the requests  into $4$ classes and picks only the one for which the multi-commodity flow solution is the highest, hence resulting in a solution of at least a $1/4$ fraction of the solution of the
     previous stage.

     Stage~\ref{step:old4} applies a randomized rounding procedure to the flows that are picked in stage~\ref{step:old3}.
     The expected size of the solution is equal to the total flow left from the previous stage,
      but the solution might not be feasible.

    Stage~\ref{step:old5} applies a filtering procedure to the solution of the previous stage,
     in order to get a feasible solution on the sketch graph.  By Claim~\ref{claim:filter},
     the expected size of this  solution is at least a  $1-O(1/k)$ fraction of the solution given
     by stage~\ref{step:old4}. Observe that $1-O(1/k)=\Omega(1)$ as $k =\Omega(\log \log  \log n)$ in any
      relevant invocation of the algorithm. We note that the proof of Claim~\ref{claim:filter} uses the upper
      bound $\dmax$ on the distance that each packet has to travel.

   Stage~\ref{step:old6} further reduces the size of the solution when the algorithm selects a subset of the requests that
   have survived so far, using a maximum flow algorithm applied to each SW quadrant. This is done in order to allow for the solution to be feasible in the original space-time grid ($\Gst$).  By Corollary~\ref{coro:limit}, the expected size of the solution after this stage is
  an $\Omega(1)$ fraction of  the expected size before this stage. We note that the proof leading to
   Corollary~\ref{coro:limit} uses the lower bound $\dmin$ on the distance that each packet has to travel.

  Stage~\ref{step:old7} gives the final routing without further reducing the size of the solution.

  We conclude that the algorithm (that we use for short, medium and long requests)
  is a randomized $O(1)$-approximation algorithm (in fact with respect to the fractional optimum).
\end{proof}

As explained at the top of Section~\ref{sec:randomized_algo}, given the original problem with capacities $B$ and $c$, we run the randomized algorithm on a modified network with both edge and buffer capacities equal to $\gamma=\min\{B,c\}$.
Since the optimal fractional solution on this modified network is only a $\max\{\gamma/B,\gamma/c\}$-fraction away from the
optimal solution on the original network, we have the following.
\begin{theorem}
\label{thm:final_randomized}
 The excepted approximation ratio of the algorithm for short, medium and long  packets is $O(1)$ when $B=\Theta(c)$.  \end{theorem}

The above theorem, together with Lemma~\ref{lemma:vshort}, concludes the proof of our main result as stated in Theorem~\ref{th:main_result}.

\bibliographystyle{alpha}
\bibliography{packet2}
\appendix
\section{Reduction of Packet-Routing to Path Packing}\label{sec:reduction}
\subsection{Space-Time Transformation}
\label{sec:spacetime}

A\emph{ space-time transformation} is a method to map schedules in a directed graph
over time into paths in a directed acyclic graph~\cite{AAF,ARSU,AZ,RR}.  Let
$G=(V,E)$ denote a directed graph.  The space-time transformation of $G$ is the
acyclic directed infinite graph $G^{st}=(V^{st},E^{st})$, where:
\begin{inparaenum}[(i)]
\item $V^{st} \triangleq V\times \mathbb{N}$.  We refer to every vertex $(v,t)$ as a
  \emph{copy} of $v$. Namely, each vertex has a copy for every time step. We often
  refer to the copies of $v$ as the \emph{row} of $v$.
\item $E^{st}\triangleq E_0\cup E_1$ where the set of forward edges is defined by
  $E_0\triangleq \{ (u,t)\rightarrow(v,t+1)\::\: (u,v)\in E~,~t\in\mathbb{N}\}$ and
  the set of store edges is defined by $E_1 \triangleq \{ (u,t)\rightarrow (u,t+1)
  \::\: u\in V, t\in \mathbb{N}\}$.
\item The capacity of every forward edge is $c$, and the capacity of every store edge
  is $B$.
\end{inparaenum}
Figure~\ref{fig:st} depicts the space-time graph $G^{st}$ for a directed path over
$n$ vertices.  Note that we refer to a space-time vertex as $(v,t)$ even though the
$x$-axis corresponds to time and the $y$-axis corresponds to the nodes.  We often
refer to $G^{st}$ as the space-time grid.

\subsection{Untilting}
The forward edges of the space-time graph $G^{st}$ are depicted in Fig.~\ref{fig:st}
by diagonal segments. We prefer the drawing of $G^{st}$ in which the edges are
depicted by axis-parallel segments~\cite{RR}. Indeed, the drawing is rectified by
mapping the space-time vertex $(v,t)$ to the point $(v,t-v)$ so that store edges are
horizontal and forward edges are vertical.  Untilting simplifies the definition of
tiles and the description of the routing.  Figure~\ref{fig:stuntilt} depicts the
untilted space-time graph $G^{st}$ (e.g., the node $(2,1)$ is mapped to $(2,-1)$.).

\begin{figure}
    \centering
\begin{subfigure}{0.38\textwidth}
\includegraphics[width=\textwidth]{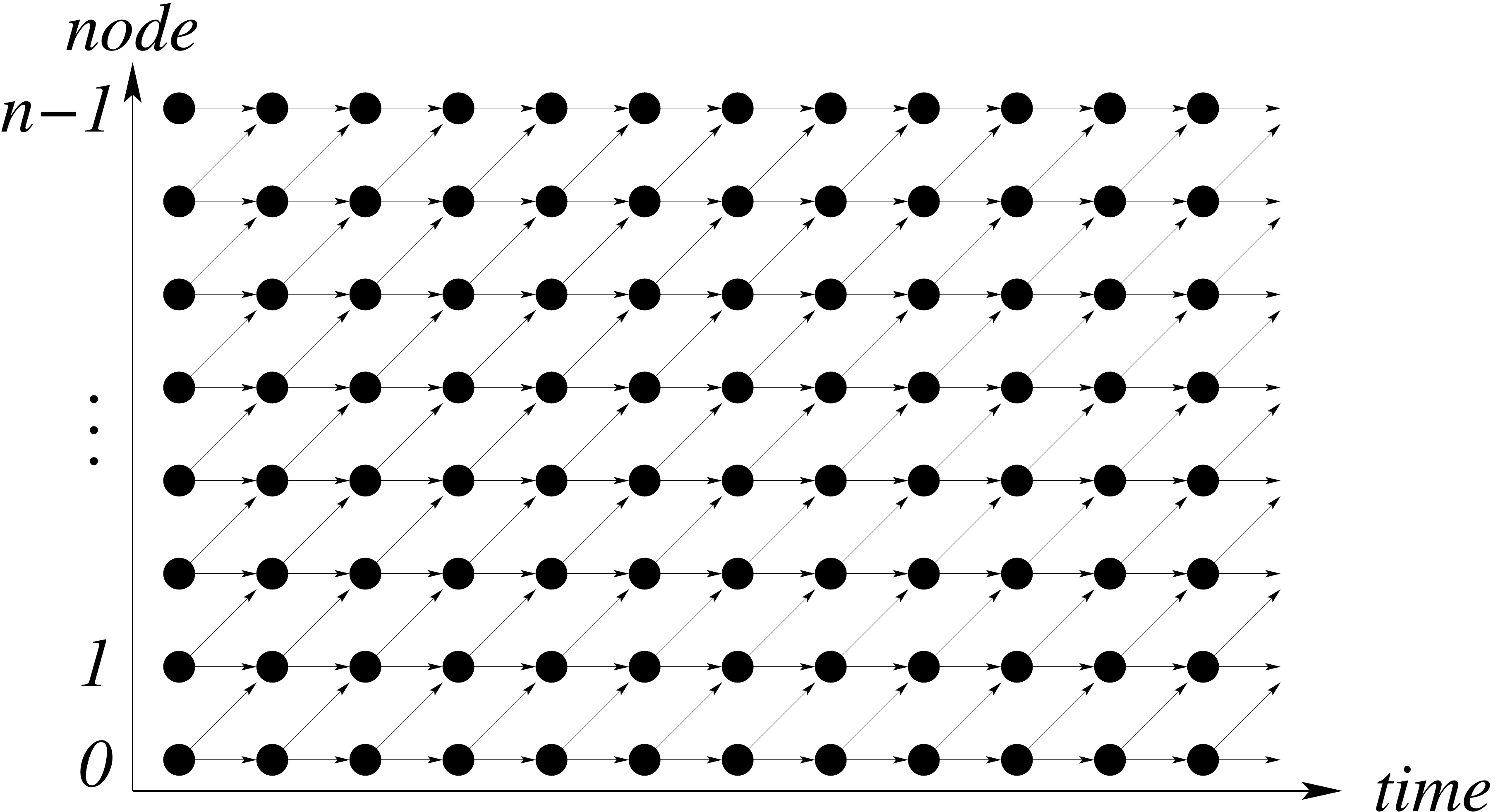}
\caption{}
\label{fig:st}
\end{subfigure}
\quad
\begin{subfigure}{0.53\textwidth}
\includegraphics[width=\textwidth]{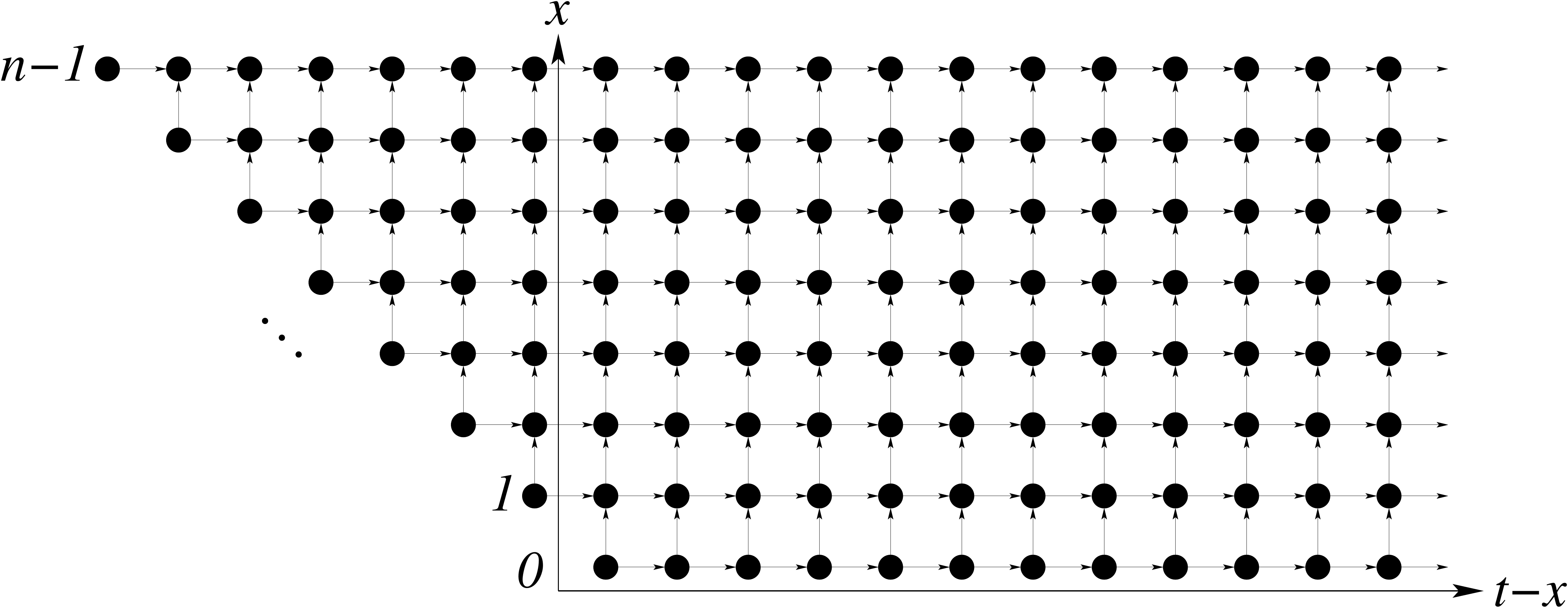}
\caption{}
\label{fig:stuntilt}
\end{subfigure}
\caption{The space-time graph $G^{st}$ before and after untilting~\cite{EM14}.}
\end{figure}

\subsection{The Reduction}
A schedule $s_i$ for a packet request $r_i$ specifies a path $p_i$ in $G^{st}$ as
follows. The path starts at $(a_i,t_i)$ and ends in a copy of $b_i$. The edges of
$p_i$ are determined by the actions in $s_i$; a store action is mapped to a store
edge, and a forward action is mapped to a forward edge. We conclude that a schedule
$S$ induces a \emph{packing} of paths such that at most $B$ paths cross every store
edge, and at most $c$ paths cross every forward edge.  Note that the length of the
path $p_i$ equals the length of the schedule $s_i$. Hence we can reduce each packet
request $r_i$ to a path request $r^{st}_i$ over the space-time graph. Vice versa, a
packing of paths $\{p_i\}_{i\in I}$, where $p_i$ begins in $(a_i,t_i)$ and ends in a
copy of $b_i$ induces a schedule.\footnote{In~\cite{AZ}, super-sinks are added to the
  space-time grid so that the destination of each path request is single vertex
  rather than a row.} We conclude that there is a one-to-one correspondence between
schedules and path packings.

\section{Multi-Commodity Flow Terminology}\label{append:MCF}
\subparagraph*{Network.}
A network $N$ is a directed graph\footnote{The graph $G$ is this section is an
  arbitrary directed graph, not a directed path. In fact, we use MCF over the space-time graph
  of the directed grid with super sinks for copies of each vertex.} $G=(V,E)$, where
edges have non-negative capacities $c(e)$.  For a vertex $u\in V$, let $\out(u)$
denote the outward neighbors, namely the set $\{y \in V \mid (u,y)\in E\}$.
Similarly, $\inn(u) \triangleq \{ x\in V \mid (x,u)\in E\}$.

\subparagraph*{Grid Network.}
A grid network $N$ is a directed graph $G=(V,E)$ where $V=[n]\times \mathbb{N}$ and
$(i,t_1)\rightarrow (j,t_2)$ is an edge in $E$ if  and only if $t_2=t_1+1$ and $0\leq j-i\leq 1$.

\subparagraph*{Commodities/Requests.} A \emph{request} $r_i$ is a pair $(a_i,b_i)$, where
$a_i\in V$ is the \emph{source} and $b_i\in V$ is the \emph{destination}. We often
refer to a request $r_i$ as \emph{commodity} $i$. The request $r_i$ is to ship
commodity $i$ from $a_i$ to $b_i$. All commodities have unit demand.

In the case of space-time grids, a request is a triple $(a_i,b_i,t_i)$ where
$a_i,b_i\in [n]$ are the source and destination, and $t_i$ is the time of arrival.
The source in the grid is the node $(a_i,t_i)$. The destination in the grid is any
copy of $b_i$, namely, any vertex $(b_i,t)$, where $t\in \mathbb{N}$ (see Section~\ref{sec:spacetime}).

\subparagraph*{Single commodity flow.}
Consider commodity $i$.  A \emph{single-commodity flow} from $a_i$ to $b_i$ is a
function $f_i:E \rightarrow \mathbb{R}^{\geq 0}$ that satisfies the following
conditions:
\begin{enumerate}[(i)]
\item Capacity constraints: for every edge $(u,v)\in E$, $0 \leq f_i(u,v)\leq c(u,v)$.
\item Flow conservation: for every vertex $u\in V\setminus\{a_i,b_i\}$
  \begin{align*}
\sum_{x\in \inn (u)} f_i(x,u) &= \sum_{y\in \out (u)} f_i(u,y).
  \end{align*}
\item Demand constraint:  $|f_i|\leq 1$ (amount of flow $|f_i|$ defined below).
\end{enumerate}
The \emph{amount} of flow delivered by the flow $f$ is defined by
\begin{align*}
  |f_i| & \triangleq \sum_{y\in \out (a_i)} f_i(a_i,y) - \sum_{x\in \inn (a_i)} f_i(x,a_i).
\end{align*}

The \emph{support} of a flow $f_i$ is the set of edges $(u,v)$ such that
$f_i(u,v)>0$. As cycles in the support of $f_i$ can be removed without decreasing
$|f_i|$,  one may assume that the support of $f_i$
is acyclic.

\subparagraph*{Multi-commodity flow (MCF).}
In a multi-commodity flow (MCF) there is a set of commodities $I$, and, for each commodity
$i\in I$, we have a source-destination pair denoted by $(a_i,b_i)$.  Consider a
sequence $F\triangleq\{f_i\}_{i\in I}$ of single-commodity flows, where each $f_i$
is a single commodity flow from the source vertex $a_i$ to the destination vertex
$b_i$.  We abuse notation, and let $F$ denote also the sum of the flows, namely $F: E
\rightarrow \mathbb{R}$, where $F(e)\triangleq \sum_{i\in I} f_i (e)$, for every edge $e$.
A sequence $F$ is a \emph{multi-commodity flow} if, in addition to the requirements defined above for each flow $f_i$, $F$ satisfies the
\emph{cumulative capacity constraints} defined by:
\begin{align*}
  \text{for every edge $(u,v)\in E$:} &~~~F(u,v) \leq c(u,v).
\end{align*}

The \emph{throughput} of an MCF $F\triangleq\{f_i\}_{i\in I}$ is defined to be
$\sum_{i\in I} |f_i|$.  In the maximum throughput MCF problem, the goal is to find an
MCF $F$ that maximized the throughput.

An MCF is called \emph{all-or-nothing}, if $|f_i|\in \{0,1\}$  for every commodity $i\in
I$.  An MCF is called \emph{unsplittable} if the support of each flow is a simple path
. An MCF is \emph{integral} if it is both all-or-nothing and
unsplittable.  An MCF that is not integral is called a \emph{fractional} MCF.

\section{Randomized Rounding Procedure}\label{sec:RR}
In this section we present material from~\cite{raghavan1987randomized} about
randomized rounding. The proof of the Chernoff bound is also based on~\cite{young1995randomized}.

Given an instance $F=\{f_i\}_{i\in I}$ of a fractional multi-commodity flow, we are
interested in finding an integral (i.e., all-or-nothing and unsplittable)
multi-commodity flow $F'=\{f'_i\}_{i\in I}$ such that the throughput of $F'$
is as close to the throughput of $F$ as possible.

\begin{observation}\label{obs:acyclic}
As flows along cycles are easy to eliminate, we
assume that the support of every flow $f_i\in F$ is acyclic.
\end{observation}

We employ a randomized procedure, called \emph{randomized rounding}, to obtain $F'$
from $F$. We emphasize that all the random variables used in the procedure are
independent.  The procedure is divided into two parts. First, we flip random coins to
decide which commodities are supplied. Next, we perform a random walk along the
support of the supplied commodities. Each such walk is a simple path along which the
supplied commodity is delivered.  We describe the two parts in details below.

\medskip
\subparagraph*{Deciding which commodities are supplied.}
For each commodity, we first decide if $|f'_i|=1$ or $|f'_i|=0$.  This decision is
made by tossing a biased coin $b_i\in\{0,1\}$ such that
\begin{align*}
  \prbig{b_i=1}&\triangleq |f_i| \leq 1.
\end{align*}
If $b_i=1$, then we decide that $|f'_i|=1$
Otherwise, if $b_i=0$, then we decide that $|f'_i|=0$.

\medskip
\subparagraph*{Assigning paths to the supplied commodities.}
For each commodity $i$ that we decided to fully supply (i.e., $b_i=1$), we assign a
simple path $P_i$ from its source $s_i$ to its destination $t_i$ by following a
random walk along the support of $f_i$. At each node, the random walk proceeds by
rolling a dice. The probabilities of the sides of the dice are proportional to the
flow amounts. A detailed description of the computation of the path $P_i$ is given in Algorithm~\ref{alg:path}.

\begin{algorithm}
  \caption{Algorithm for assigning a path $P_i$ to flow $f_i$.}
\label{alg:path}
 \begin{algorithmic}[1]
\State $P_i\gets \emptyset$.
\State $u\gets s_i$
\While {$u\neq t_i$} \Comment{did not reach $t_i$ yet}
\State $v\gets \textit{choose-next-vertex}(u)$.
\State Add  $(u,v)$ to $P_i$
\State $u\gets v$
\EndWhile
\State \textbf{return} $(P_i)$.

   \Procedure{$\textit{choose-next-vertex}$}{$u,f_i$}

   \State Let $\out (u,f_i)$ denote the set of edges in the support of $f_i$ that
   emanate from $u$.

   \State Consider a dice $C(u,f_i)$ with $|\out (u,f_i)|$ sides.
   The side corresponding to an edge $(u,v)\in \out(u,f_i)$ has probability
   $f_i{(u,v)} / (\sum_{(u,v') \in \out(u,f_i)} f_i(u,v'))$.

   \State Let $v$ denote the outcome of a random roll of the dice $C(u,f_i)$.

   \State \textbf{return $(v)$}
\EndProcedure
  \end{algorithmic}
\end{algorithm}

\medskip
\subparagraph*{Definition of $F'$.}
Each flow $f'_i\in F'$ is defined as follows. If $b_i=0$, then $f'_i$ is identically
zero. If $b_i=1$, then $f'_i$ is defined by
\begin{align*}
  f'_i(u,v) &\triangleq
  \begin{cases}
    1 & \text{if $(u,v)\in P_i$,}\\
    0& \text{otherwise.}
  \end{cases}
\end{align*}
Hence, $F'=\{f'_i \mid b_i=1\}$ is an all-or-nothing unsplittable flow, as required.

\subsection{Expected flow per edge}
The following claim can be proved by induction on the position of an edge in a
topological ordering of the support of $f_i$.
\begin{claim} \label{claim:expec}
  For every commodity $i$ and every edge $(u,v)\in E$:
  \begin{align*}
    \prbig{(u,v)\in P_i} &= f_i(u,v),\\
    \expec{f'_i(u,v)} &= f_i(u,v).
  \end{align*}
\end{claim}

\end{document}